\documentclass[12pt]{article}
\usepackage{amsfonts,amssymb,amsmath,amsthm,enumerate,mathrsfs,tikz,xparse}
\usepackage{tikz}
\usepackage{tkz-graph}
\usepackage{tkz-berge}
\usepackage{comment}

\usepackage{graphicx}
\usepackage[T1]{fontenc}

\usepackage{algorithm}
\usepackage{algpseudocode}
\usepackage{float}
\usepackage{tikz}
\usepackage{amssymb}

\usepackage{mathtools}
\usepackage{verbatim}

\usepackage{multirow}
\usepackage{longtable}

\newtheorem{theorem}{Theorem}
\newtheorem{prop}{Proposition}
\newtheorem{lemma}{Lemma}
\newtheorem{obs}{Observation}

\newtheorem{cor}{Corollary}

\newtheorem{remark}{Remark}
\long\def\comment#1\endcomment{}
\date{}

\begin{document}

\title{A polynomial-time approximation to a minimum dominating set in a graph}

\author{Frank Angel Hern\'andez Mira $^{(1)}$, Ernesto Parra Inza $^{(2)}$, 
\\ Jose Mar\'ia Sigarreta $^{(3)}$ and Nodari Vakhania $^{(2)}$ }

\maketitle

\begin{abstract}
A {\em dominating set} of a graph $G=(V,E)$ is a subset of vertices $S\subseteq V$ such that every vertex $v\in V\setminus S$ has at least one neighbor in $S$. Finding a dominating set with the minimum cardinality in a connected graph $G=(V,E)$ is known to be NP-hard. A polynomial-time  approximation algorithm for this problem, described here, works in two stages. At the first stage a dominant set is generated by a greedy algorithm, and at the second stage this dominating set is purified (reduced). The reduction is achieved by the analysis of the flowchart of the algorithm of the first stage and a special kind of clustering of the dominating set generated at the first stage. The clustering of the dominating set naturally leads to a special kind of a spanning forest of graph $G$, which serves as a basis for the second purification stage.  We expose some types of graphs for which the algorithm of the first stage already delivers an optimal solution and derive sufficient conditions when the overall algorithm constructs an optimal solution. We give three alternative approximation ratios for the algorithm of the first stage, two of which are expressed in terms of solely invariant problem instance parameters, and we also give one additional approximation ratio for the overall two-stage algorithm. The greedy algorithm of the first stage turned out to be essentially the same as the earlier known state-of-the-art algorithms for the set cover and dominating set problem Chv\'atal \cite{chvatal} and Parekh \cite{parekh}. The second purification stage results in a significant reduction of the dominant set created at the first stage, in practice. The practical behavior of both stages was verified for randomly generated problem instances. The computational experiments emphasize
the gap between a solution of Stage 1 and a solution of Stage 2.
\end{abstract}


\small $^{(1)}$ Centro de Ciencias de Desarrollo Regional, UAGro, \emph{fmira8906@gmail.com;}\\
\small $^{(2)}$ Centro de Investigaci\'on en Ciencias, UAEMor, \emph{eparrainza@gmail.com}, \emph{nodari@uaem.mx}
\small $^{(3)}$ Facultad de Matem\'aticas, UAGro, \emph{josemariasigarretaalmira@hotmail.com}\\

\section{Introduction}

{\bf Problem description.}
One of the most studied problems in combinatorial optimization and graph theory
are covering and partitioning problems in graphs. A subset of vertices in a graph is a 
dominating set if every vertex of that graph which is not in that subset has at 
least one neighbor in that subset. More formally, given a simple connected undirected 
graph $G =(V,E)$ with $|V| = n$ vertices and $|E|= m$ edges, a set $S\subseteq V$ is 
called a \emph{dominating set} of that graph 
if for all $v\in V$ either $v\in S$ or there exists a vertex $u$ in $V\setminus S$ 
such that edge $(v,u) \in E$ (without loss of generality, we assume
that graph $G$ is connected). A widely studied such problem is the dominating set problem. The {\sc Minimum Dominating Set} problem seeks for a dominating set
with the minimum cardinality. 

Given a vertex $v\in V$, $N(v)$ is the set of neighbors 
or the open neighbourhood of $v$ in $G$; that is, 
$N(v)=\{u\in V: (u,v)\in E\}$. We denote by $\delta_G(v)=|N(v)|$ 
the \emph{degree} of vertex $v$ in $G$ (we may omit the argument $G$ when this will
cause no confusion). We let $\delta_G = \min_{v\in V}\{\delta(v)\}$ and 
$\Delta_G = \max_{v\in V}\{\delta(v)\}$. The \emph{private neighborhood} 
of vertex $v \in S\subseteq V$ is defined by $\{u \in V: N(u)\cap S = \{v\}\}$; 
a vertex in the private neighborhood of vertex $v$ is said to be its 
\emph{private neighbor} with respect to set $S$.

The \emph{domination number} of graph $G$, denoted as $\gamma(G)$, 
is the minimum cardinality 
of a dominating set for that graph; we shall refer to  a corresponding dominating set of cardinality $\gamma(G)$ as a $\gamma(G)$-set. A dominating set is {\em minimal} if by removing any of its elements the resultant reduced set becomes non-dominating. In fact, it is not so difficult to construct  a polynomial-time algorithm that generates a minimal dominating set. However,  not  necessarily a minimal dominating set approximates  well a minimum dominating set, i.e., given a minimal dominating set, there may exist a 
non-minimal dominating set with a (much) smaller number of vertices. 

The reader is referred to Haynes et al. \cite{haynes1} 
for further details on the basic graph terminology. 
The problem of domination in graphs was mathematically formalized by Berge  
and Ore in \cite{berge-1962} and \cite{ore}, respectively. Currently, this topic 
has been detailed in the two well-known books by 
Haynes, Hedetniemi, and Slater in \cite{haynes1} and \cite{haynes2}. 
The  theory of domination in graphs is an area of 
increasing interest in discrete mathematics and combinatorial computing, 
in particular, because of a number of important real-life applications. 
Such applications arise, for instance, during the analysis 
of social networks (see  Kelleher and Cozzens in \cite{kelleher}), in efficient 
identification of web communities (see Flake et al. in \cite{flake}), in 
bioinformatics (see Haynes et al. in \cite{haynes3}) 
and food webs (see Kim et al. in \cite{kim}), in the study 
of transmission of information in networks associated with defense systems
(see Powel in \cite{powel}), also in distributed computing (see Kuhn et al. in 
\cite{kuhn1}, Kuhn and Wattenhofer in \cite{kuhn2}). 
The variations of the domination problem and their applications have been 
widely studied, see for example, \cite{haynes1} and \cite{haynes2}. 

\smallskip

{\bf Related work.}
There are a few known results that bound the size of a minimum dominating set.
Given graph $G$, let $\Delta(G)$ be the maximum 
vertex degree in graph $G$, and let $d(G)$ be the \emph{diameter} of graph $G$, {\em i.e.}, the 
maximum number of edges on a shortest path between any pair of vertices in that graph.
Then  $\frac{n}{\Delta+1} \leq \gamma(G) 
\leq \frac{n}{2}$ or alternatively, $\frac{d(G)+1}{3} \leq \gamma(G) \leq n-\Delta$
(see for instance \cite{haynes1}). Despite these bounds, it is NP-hard
to find a minimum dominating set in a graph  and 
remains such for planar graphs  with maximum vertex degree 3 and for 
regular planar graphs with the maximum vertex degree 4, see  Garey and Johnson  
\cite{GJ}. It is also NP-hard to achieve a $c\log(\Delta)$-approximation, for 
a constant $c > 0$ Raz and Safra \cite{raz}. In terms of $n$, an optimal solution
cannot be approximated within a factor strictly smaller than $\ln(n)$ 
(unless $P=NP$), see more in \cite{chebik}. 

Exact (exponential-time) algorithms have been suggested for the problem. A branch 
and reduce algorithm with a measure and conquer approach for the dominating set problem
was proposed in Fomin et al. \cite{formin}. The measure and conquer approach helps 
to obtain good upper bounds on the running time for using branch and reduce method, 
often improving upon the currently best known bounds. It has been used successfully 
for the dominating set problem, see Fomin et al. \cite{formin}, when was obtained a running time of $O(2^{0.610n})$. Van Rooij and 
Bodlaender in \cite{rooij}, have obtained an $O(1.4969^n)$ time 
and polynomial space  algorithm using a measure and conquer approach, this improving the bound obtained in \cite{formin}.  
$O(1.4173^n)$, $O(1.4956^n)$ and $O(1.2303^{(1+\sqrt{1-2c})n})$ time algorithms 
for finding a minimum dominating set for some graph classes (chordal graphs, circle 
graphs and dense graphs, respectively) have been proposed in Gaspers et al. \cite{gaspers}. 

Alternatively to the exact exponential-time algorithms that cannot give a solution for
already  moderate sized practical instances, greedy algorithms have  been  proposed. 
Eubank et al. \cite{eubank} and Campan et al. \cite{campan} present heuristic 
algorithms with the experimental study of their practical behavior. The authors
do not specify the approximation factor and the time complexity of the proposed
algorithms. 

Chv\'atal in \cite{chvatal} describes an approximation algorithm with a guaranteed
approximation ratio $\ln\left(\frac{n}{OPT}\right) + 1 + \frac{1}{OPT}$ for the 
weighted set cover problem, where $OPT$ is the optimal objective value.  At each 
iteration, the algorithm selects vertex $s$ that minimizes the 
$\frac{w_s}{f(C\cup \{s\}) - f(C)}$, where $w_s$ is the weight of vertex $s$ and 
$f(C) = |\cup_{s\in C} N[s]|$, until $f(C) = |S|$. This algorithm can be
used to solve the minimum dominant set problem. In particular, 
Parekh \cite{parekh} has applied this greedy algorithm for the domination problem,  
and  showed that the cardinality of the dominating set created by that algorithm is 
upper bounded by $n+1-\sqrt{2m+1}$ (this algorithm turned out 
to be essentially the same as the first algorithm that we propose here). 

Exact polynomial-time algorithms and approximation schemes exist also for 
some restricted families of graphs. In particular, polynomial-time approximation 
schemes for unit disk graphs and growth-bounded graphs (the problem remains NP-hard
for such graphs) were described in \cite{huntIII} and \cite{nieberg}, respectively. Erlebach 
and van Leeuwen \cite{erlebach} have conducted a study of the problem for disk graphs 
and proposed an almost 3-approximation polynomial-time algorithm for disk graphs. 
Alternatively, an $2^{\log^*n}$-approximation polynomial-time algorithm for dominating 
set for disk graphs was proposed by Gibson and Pirwani \cite{gibson}.
There are known results for  series-parallel graphs, that is, graphs with two distinguished 
vertices called terminals, formed recursively by two simple composition operations 
(see \cite{schoenmakers}). 
Such graphs can be used to model series and parallel electric circuits. For 
series-parallel graphs, a linear-time algorithm was 
proposed by Kikuno et al. \cite{kikuno}. The so-called $k$-degenerate graphs were 
introduced by Lick and White  \cite{lick}.
Two linear time algorithms for finding a dominating set of fixed size in degenerated 
graphs are described in Alon and Gutner  \cite{alon}.
A {\em subdivision} of a graph $G$ is a graph that can be obtained from $G$ by a 
sequence of edge subdivisions. If a subdivision of graph $H$ is a subgraph of 
graph $G$, then $H$ is a {\em topological minor} of $G$. For black and white graphs that do not 
contain $K_m$ as a topological minor, the authors in \cite{alon} give an 
algorithm with running time $(O(m))^{mk}n$, whereas for $K_m$-minor-free
graphs, the running time of the algorithm is shown to be $(O(\log m))^{mk/2}n$, 
where $m$ and $k$ are fixed parameters. 
For trees,  Cockayne et al. \cite{cockayne} proposed a fast linear-time exact 
algorithm.  Siebertz \cite{siebertz} proposed a greedy algorithm approximating dominating sets on biclique-free graphs with the approximation ratio of 
$O(t\ln(\gamma(G)))$ for graphs without $K_{t,t}$ as a subgraph induced. 
This algorithms, however, do not guarantee a good approximation for general graphs.

The so-called {\em cograph} was introduced by Lerchs in \cite{lerchs}. A cograph is defined recursively as follows: $(i)$ $K_1$ is a cograph; $(ii)$ if $G$ is a cograph, then $\overline{G}$ is a cograph, and $(iii)$ if $G_1$ and $G_2$ are cographs, then $G_1 \cup G_2$ is a cograph, where $G_1=(V_1,E_1)$, $G_2=(V_2,E_2)$ and $G_1 \cup G_2 = (V_1 \cup V_2, E_1 \cup E_2)$. Golumbic and Rotics in \cite{golumbic} show that the cographs are exactly the graphs of clique-width at most 2, hence the {\sc Minimum Dominating Set} problem can be solved in polynomial time for this type of graphs. 

A dominating set $D$ of a graph $G$ is called a secure dominating set of $G$ if for every $u \in V\setminus D$ there exists $v \in D$ such that $uv\in E$ and $(D\setminus \{v\})\cup \{u\}$ is a dominating set of $G$. Recently, an $O(n+m)$-time algorithm that finds the secure domination number of a cograph was described by Pradhan et al. \cite{pradhan}. The reader is referred to Vazirani \cite{vazirani} for a more extensive survey of approximation algorithms for different types of graphs. 

Some generalizations of the minimum dominating set problem have also been studied.
For instance, in the so-called minimum 
weight dominating set problem (MWDS), each vertex is associated with a 
positive integer, and the goal is to find a dominating set with the smallest weight. 
A local search algorithm for MWDS was proposed by Wang Y. \emph{et al}. \cite{wang}.  
They designed a new fast construction procedure with four reduction rules to 
reduce the size of graphs, and developed a configuration checking strategy to improve 
the search. Guha and Khuller in \cite{guha}, proposed two heuristic
algorithms to approximate a connected dominating set (a dominating set for which  
the subgraph induced by vertices of that set is connected). Another  generalization 
of the domination problem is the $k$-domination problem that looks for the so-called
$k$-dominating set of vertices, a set of vertices such that every vertex from the
complement of that has at least $k$ neighbors from the former set. Foerster in 
\cite{k_dominating_set} proposed an $O(n(n+m))$ algorithm with approximation ratio 
$ln(\Delta+k)+1$ (we carry out a comparison with our algorithm in Section 3).

\smallskip 

{\bf Our contributions.} 
In this paper we propose a two-stage solution method 
for {\sc Minimum Dominating Set} problem yielding an overall polynomial-time 
approximation algorithm. A dominating set, which is not necessarily optimal
(and even minimal), is created at the first stage, which is reduced at the 
second purification stage. 
As already noted, the greedy algorithm of stage 1, though developed
independently from the earlier mentioned algorithm  from \cite{chvatal} and 
\cite{parekh}, turned out to be essentially the same as the latter one. 
At stage 2, the reduction of the dominating set of stage 1  is achieved 
based on the analysis of the flowchart of the algorithm of stage 1 
represented as a special kind of a spanning forest of the original graph 
$G$. This forest is formed by a special kind of clustering of  the dominating 
set formed at stage 1. Then the purification of this dominating set 
is accomplished at stage 2 using a special kind
of traversal of each rooted tree from the forest.

Thus the purification is based on the clustering of the flowchart of the algorithm of 
stage 1. The clustering  can be accomplished in different ways. Different clustering 
methods potentially yield a collection of different rooted trees. Likewise, 
different traversal and purification methods can be developed for a chosen 
cluster structure. The same traversal and purification method will 
potentially give different outcome for a given cluster structure. Thus
different combinations of a particular clustering and traversal/purification methods 
result in different overall purification algorithms of different efficiency. 
In this paper, we propose a simple clustering and traversal/purification methods
yielding an overall $O(n^2\log n)$ time procedure, whereas stage 1 is 
implemented in time $O(n^3)$ (a more substantial study of different clustering
and purification methods could well be the subject of a future research).

We establish some types of graphs for which the algorithm of stage 1 already delivers 
a minimum dominating set, and we derive sufficient conditions when the overall 
algorithm delivers a minimum dominating set. 

The behavior of our approximation algorithm of stage 1 is reflected by the 
three alternative approximation ratios that we give here. These
bounds similarly apply to the algorithm \cite{chvatal} and \cite{parekh}. 
The first two approximation ratios $ln(\frac{n}{\gamma(G)}) + 1$  and 
$ln(\Delta + 1) + 1$ are derivations from the earlier known ratios. The first
one from \cite{chvatal} is also good for our greedy algorithm, and the second one
is a special case of the above mentioned ratio $ln(\Delta+k)+1$ for the $k$-domination
problem from \cite{k_dominating_set} for $k=1$. We also show that the greedy algorithm 
of stage 1 has an alternative approximation ratio $\frac{\Delta+1}{2}$ (which 
again apply to the algorithm from \cite{chvatal} and \cite{parekh}). Note that the 
latter two bounds involve solely invariant problem instance parameters so that they 
can be easily calculated a priory before the algorithm is actually applied. 
We give one additional approximation ratio $\frac{\Delta+1}{2}$ 
for the overall two-stage algorithm, which is guaranteed to be held if 
$1\leq \Delta \leq 4$ or $n \geq \gamma(G)e^{\frac{\Delta-1}{2}}$.

The second purification stage results in a significant reduction of the dominant 
set created 
at the first stage, in practice. The practical behavior of both stages was verified 
for randomly generated problem instances. The computational 
experiments emphasize
the gap between a solution of Stage 1 and a solution of Stage 2.

In the next section we describe the greedy algorithm of stage 1.
We also give some classes of graphs for which stage 1 delivers an 
optimal solution. In Section 3 we study some useful properties of the dominating set
delivered by stage 1, and we discuss our clustering and purification 
approaches. Then we present our simple clustering and purification methods. 
In Section 5 we carry out a comparative study of the performance of our algorithm vs 
the state-of-the-art approximation algorithms for the domination and $k$-domination 
problems. In Section 6 we report our experimental results.

\section{Stage 1: The Greedy Algorithm}

In this section we give a description of the greedy algorithm of stage 1 (this 
essentially describes the greedy algorithm from \cite{chvatal} and \cite{parekh}; 
from here on, we may refer to either of the algorithms as {\em Greedy}).
Let $ G = (V, E) $ be a graph; in particular $ V $ and $ E $ will be also denoted as $ V (G) $ and $ E(G) $. For any $ U\subseteq V  $, let $ G[U] $ denote the subgraph of $ G $ induced by vertices of $ U $. For any $ U\subseteq V  $, let 
 $ \overline{U}=V \setminus U $.  For any $ U\subseteq V  $, let $ N(U) = \{v \in V\setminus U : v $ has a neighbor in $ U\} $; in particular, if $ U = {u} $, then let us simply write $ N(u) $.
Greedy works in a number of iterations limited from the above by $n$. At each iteration 
$h\geq 1$, one specially selected vertex, denoted by $v_h$, is added to the dominant
set $S_{h-1}$ of the previous iteration, i.e., we let $S_h:=S_{h-1}\cup\{$$v_h$$\}$, 
where initially we let $S_0=\emptyset$. 

At each iteration $ h $, the \textit{active degree} of a
vertex $ v \in \overline{S_{h-1}}  $ is given by $ | N(v) \setminus [S_{h-1} \cup N(S_{h-1})]| $. Then, at each iteration $ h $, vertex
$ v_h $ is any vertex of $ S_{h-1} $ with maximum active degree. 
Note that $v_1$ is a 
vertex with the maximum degree in graph $G$.  
The algorithm halts when $ S_h $ is a dominating set of $ G $. At that iteration, all 
uncovered vertices with active degree 0 (if any) are included in the set $S_h$.

\medskip

\begin{algorithm*}
\caption{Algorithm Stage-One (Greedy)}\label{alg_basic}
\begin{algorithmic}
\State Input: A graph $G$.
\State Output: A dominating set $ S $ of $ G $.
\State $h := 1$; 
 
\State $S_0:= \emptyset$;  

\{ iterative step \} 

\While{$S_h$ is not a dominating set of  graph $G$} 
\State $h := h+1$; 
\State $v_h$ := any vertex of $ \overline{S_{h-1}} $ with maximum active degree;  
\State $S_h := S_{h-1}\cup \{$$v_h$$\}$;
\EndWhile
	
\end{algorithmic}

\end{algorithm*}

\medskip

{\bf Implementation.} We represent graph $G=(V,E)$ by the
adjacency matrix $M=M(G)$, in which the entry $(i,j)$ is 1 if vertices
$v_i$ and $v_j$ are neighbors in graph $G$. The adjacency matrix  
is updated at each iteration $h$ according to the selection of vertex $v_h$. 
Initially, at the beginning of iteration 0,  $M_0 := M$. Iteratively, some 1 
entries from matrix $M_{h-1}$ of iteration $h-1$ are replaced by 0 entries in 
matrix $M_h$ of iteration $h\ge 1$. In particular, the updated matrix $M_h$ 
is obtained by replacing an entry 1 in the column corresponding to each neighbor 
of vertex $v_h$ by 0, and replacing all 1 entries in the row and the column 
corresponding to vertex $v_h$ again by 0. In this way, all edges incident to an 
already covered neighbor of vertex $v_h$ and all edges associated with vertex 
$v_h$ will be ignored from iteration $h$.\\

\begin{lemma}\label{degree}
At every iteration $h$, vertex $v_h$ with the maximum active degree is determined in time $O(n^2)$.
\end{lemma}

\begin{proof}
First note that, in every iteration $h$, for every vertex in $\overline{S}_h$ 
the active degree is calculated by summing up all the 1 entries in the row 
corresponding to that vertex. Since each row contains $n$ entries and there are
less than $n$ rows to verify, the calculation of the active degrees of all the
elements in  $\overline{S}_h$ and the selection of the maximum takes time $O(n^2)$.
\end{proof}

\begin{theorem}\label{time-1}
Algorithm Stage-One runs in time $ O(n^3) $. 
\end{theorem}

\begin{proof}
First, recall that the total number of the iterations 
(from the external loop)
in Greedy is bounded from the above by $n$. For each iteration $h$, the
selection of vertex $v_h$ takes time $O(n^2)$ (Lemma \ref{degree}). The cost of
the update of the adjacency matrix at iteration $h$  is again bounded by $O(n^2)$. 
Hence, Greedy finds a dominating set $S$ in time $O(n^3)$.
\end{proof}

It is straightforward to verify that the above estimation also holds for the 
algorithm from \cite{parekh}. 
The following remark gives a sufficient optimality condition for the algorithm Greedy Algorithm (and also that from \cite{parekh}), 
and will also be useful in the estimation of the approximation ratio of the
proposed overall algorithm.\\

\begin{remark}\label{obs-1}
If $|S| \leq 2$ then $S$ is a minimum dominating set.
\end{remark}

\begin{proof}
It clearly suffices to consider only the two cases, $\gamma(G)=1$ and $\gamma(G)=2$. 
Indeed, suppose, first, that $\gamma(G)=1$. Then there exists vertex $v\in V$ adjacent
to every other vertex in $V$. Then at the initial iteration 0 of Greedy Algorithm 
the (active) degree of vertex $v$  would be the maximum and this vertex will be 
selected as $v_0$, and, $S:=\{ v\}$ (among all such vertices, ties can 
clearly be broken arbitrarily). Otherwise, $\gamma(G)\ge 2$, 
whereas the dominating set returned by stage 1 contains no more than 2
elements, hence it must be optimal.
\end{proof}

Figure \ref{fig_1992} below shows an example where the greedy 
algorithm of Stage 1 returns a non-minimal dominating set with cardinality 
greater than $n/2$.

 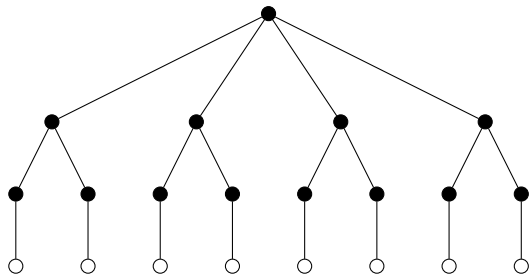
\begin{figure}[h]
	\centering
	\begin{tikzpicture}[scale=.48, transform shape]
				
		\node [draw, shape=circle] (r1) at  (6,-3) {};
		\node [draw, shape=circle] (r2) at  (8,-3) {};
		\node [draw, shape=circle] (r3) at  (10,-3) {};
		\node [draw, shape=circle] (r4) at  (12,-3) {};
		\node [draw, shape=circle] (r5) at  (14,-3) {};
		\node [draw, shape=circle] (r6) at  (16,-3) {};
		\node [draw, shape=circle] (r7) at  (18,-3) {};
		\node [draw, shape=circle] (r8) at  (20,-3) {};
		
		\node [draw, shape=circle, fill=black] (p1) at  (6,-1) {};
		\node [draw, shape=circle, fill=black] (p2) at  (8,-1) {};
		\node [draw, shape=circle, fill=black] (p3) at  (10,-1) {};
		\node [draw, shape=circle, fill=black] (p4) at  (12,-1) {};
		\node [draw, shape=circle, fill=black] (p5) at  (14,-1) {};
		\node [draw, shape=circle, fill=black] (p6) at  (16,-1) {};
		\node [draw, shape=circle, fill=black] (p7) at  (18,-1) {};
		\node [draw, shape=circle, fill=black] (p8) at  (20,-1) {};
		
		\node [draw, shape=circle, fill=black] (q1) at  (7,1) {};
		\node [draw, shape=circle, fill=black] (q2) at  (11,1) {};
		\node [draw, shape=circle, fill=black] (q3) at  (15,1) {};
		\node [draw, shape=circle, fill=black] (q4) at  (19,1) {};
		
		\node [draw, shape=circle, fill=black] (x1) at  (13,4) {};
		
		\draw(r1)--(p1)--(q1)--(x1);
		\draw(r2)--(p2)--(q1);
		\draw(r3)--(p3)--(q2)--(x1);
		\draw(r4)--(p4)--(q2);
		\draw(r5)--(p5)--(q3)--(x1);
		\draw(r6)--(p6)--(q3);
		\draw(r7)--(p7)--(q4)--(x1);
		\draw(r8)--(p8)--(q4);
		
	\end{tikzpicture}
	\caption{\small The solution delivered by Greedy Algorithm is represented 
	by the dark vertices. }\label{fig_1992}
\end{figure}

\subsection{Favorable and unfavorable graph subclasses}

In this subsection we give the examples of graph families for which Greedy Algorithm 
gives an optimal solution (i.e., its output is a $\gamma(G)$-set). Also, we give an example of a graph for which the algorithm fails to generate an optimal solution. To shows the first aim, we define the following operations on the edges of graph $G$.

{\bf\emph{Double Subdivision with inflation of size $k$:}} Given an edge $uv$, remove this edge, and add the following $2k$ vertices $w_{11}, \ldots, w_{1k}$ and $w_{21}, \ldots, w_{2k}$. Finally, add the edges $uw_{1i}$, $w_{1i}w_{2i}$ and $w_{2i}v$ for every $i \in \{1,\ldots,k\}$.

{\bf\emph{Addition of $t$ pendant vertices:}} Given a vertex $x$ add $t$ new vertices $y_1, y_2, \ldots, y_t$ and the edges $xy_i$, for every $i\in \{1,\ldots,t\}$.

For any graph $H$ and for any positive integers $a,b$, let $W_{a,b}(H)$ denote the graph obtained from $H$ as follows:

\begin{enumerate}

\item Apply the operation {\bf Double Subdivision with inflation of size $k_i$} for each edge of $H$ with $1 \leq k_i \leq a$ and $1\leq i \leq m$, notice that for each edge we can get different values for $k_i$.

\item Apply the operation {\bf Addition of $t_i$ pendant vertices} for each vertex of $H$ with $0 \leq t_i \leq b$ and $1\leq i \leq n$, $t_i$ can take different values as above.
 
\end{enumerate}

Let now $\mathcal{F}_1$ be the family of graphs $W_{a,b}(H)$ for any graph $H$ and any positive integers $a,b$. It is easily observed that for any graph belonging to $\mathcal{F}_1$ Greedy Algorithm gives a minimum dominating set. A fairly representative graph of the family $\mathcal{F}_1$ is depicted in Figure \ref{fig2}, where $H=C_3$.

\begin{figure}[h]
\centering
\begin{tikzpicture}[scale=.55, transform shape]

\node [draw, shape=circle, fill=black] (u) at  (0,0) {};
\node [draw, shape=circle, fill=black] (v) at  (5,5) {};
\node [draw, shape=circle, fill=black] (w) at  (10,0) {};

\node [draw, shape=circle] (s1) at  (3,1) {};
\node [draw, shape=circle] (s2) at  (3,-1) {};
\node [draw, shape=circle] (s3) at  (7,1) {};
\node [draw, shape=circle] (s4) at  (7,-1) {};

\node [draw, shape=circle] (s5) at  (10,2) {};
\node [draw, shape=circle] (s6) at  (8,2) {};
\node [draw, shape=circle] (s7) at  (9,4) {};
\node [draw, shape=circle] (s8) at  (7,3.5) {};

\node [draw, shape=circle] (s9) at  (2,2) {};
\node [draw, shape=circle] (s10) at  (0,2) {};
\node [draw, shape=circle] (s11) at  (3,3.5) {};
\node [draw, shape=circle] (s12) at  (1,4) {};

\node [draw, shape=circle] (l1) at  (3,7) {};
\node [draw, shape=circle] (l2) at  (7,7) {};

\node [draw, shape=circle] (l3) at  (-2,0) {};
\node [draw, shape=circle] (l4) at  (0,-2) {};

\draw(u)--(s1)--(s3)--(w);
\draw(u)--(s2)--(s4)--(w);

\draw(w)--(s5)--(s7)--(v);
\draw(w)--(s6)--(s8)--(v);

\draw(u)--(s9)--(s11)--(v);
\draw(u)--(s10)--(s12)--(v);

\draw(v)--(l1);
\draw(v)--(l2);

\draw(u)--(l3);
\draw(u)--(l4);

\end{tikzpicture}
\caption{\small A graph in $W_{2,2}(C_3) \subseteq \mathcal{F}_1$, where the three bold vertices 
 form the output of the algorithm which is optimal}\label{fig2}
\end{figure}
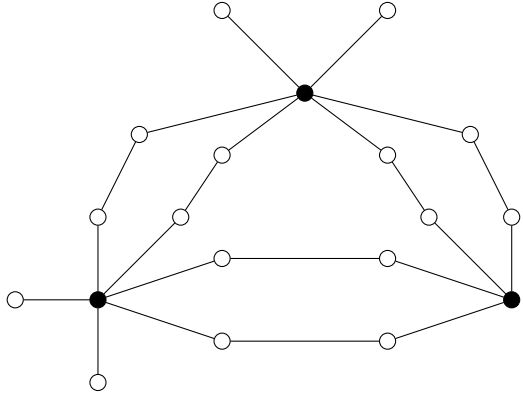

\section{Flowchart of Greedy Algorithm - clusters and induced forests}

It is natural to think of a purification procedure that eliminates some possibly 
redundant vertices from set $S$ (which is not necessarily a minimum 
and even a minimal dominating set). At Stage 2 we reduce the dominant set 
$S$ of Stage 1 to a subset $S^* \subseteq S$  which remains
dominating, is minimal and contains at most $n/2$ vertices. The remaining part 
of the paper is  dedicated basically to Stage 2. 

We will refer to a vertex $x\in S$ as {\em purified} if it is eliminated 
from set $S$ at stage 2. The purified vertices are determined in a specially 
formed subgraphs of graph $G$, the so-called {\em clusters}. Each cluster is
a subgraphs of graph $G$ induced by a subset of vertices of set $S$. 

We treat a cluster $C$ as a rooted tree $T(C)$, and we 
form a set of clusters while Greedy runs for a particular
problem instance. The union of all clusters define a forest $T$, a set of 
rooted trees $\{T_1,\dots, T_k\}$, where $k$ is the total number of the
clusters.  

The purification procedure, that we describe a bit later, invoked for 
a cluster $C_i$, works on a number of iterations. We will denote by $T^h_i$ 
the purified tree corresponding to iteration $h$, and by $T^\prime_i$ the 
purified tree constructed for cluster $C_i$. $T^\prime$ will denote 
the overall purified forest, a result of
$k$ independent calls to the purification procedure. 

Now we introduce some necessary definitions that will permit us to define 
formally our  clusters. 
 
Let $ S $ be the output of Algorithm Stage-One and let
 $ S = \{v_1, \dots , v_p\} $, for some integer $ p $, where $ v_h $
is the vertex included at iteration $h$.
Let $ S(h) $ be the set of vertices uncovered by $ v_1, \dots , v_{h-1} $ and covered by $ v_h $ 
(note that sets $ S(1), \dots , S(p) $ are mutually disjoint).

Let us say that a \textit{tied pair} is an element of the set 
$ P_S = \{($$v_h$$,v): $$v_h$$ \in S, v \in S \bigcap S(h), \; h \in \{1, \dots , p\}\}$; 
note that each tied-pair is an edge of $ G $. The following proposition 
immediately follows from this definition.\\

\begin{prop}\label{tight}
Let $ v_b, v_c \in S $, with $ b < c $, be such that $ (v_b, v_c) \in E$. Then one has
$ (v_b, v_c) \in P_S  $ unless $ (v_a, v_c) \in P_S  $ for some $ v_a \in S  $ with $ a < b $.	
\end{prop}

Let $ G(P_S) $ denote the subgraph of $ G $ induced by the pairs [i.e. by the edges of $ G $] in
$ P_S $. That is, vertices of $ G(P_S) $ are those involved in edges in $ P_S $, and edges of $ G(P_S) $
are those in $ P_S $.

Note that $ G(P_S) $ is a forest of rooted trees say $ T_1, . . . , T_k $ for some integer $ k $ : in fact, no
cycle say $ (v_a, v_b),(v_b, v_c), . . . ,(v_h, v_a) $ may arise in $ G(P_S) $, since by the above definitions
one would have $ a < b < c < . . . < h < a $ (a contradiction). In particular the root of $ T_i $
(for $ i = 1, . . . , k $) is the vertex of $ T_i $ which was added to the greedy solution [i.e. to $ S $]
at the earliest iteration [of Algorithm Stage-One] over vertices of $ T_i $.\\

{\bf Example.} From a first glance, one may suggest that Greedy  
delivers an optimal dominating set if  there arises no tied pair during its 
execution, i.e., $P_S=\emptyset$. Although for the vast majority of the 
graphs this is true,  there are graphs for which this is not true. 
As an example, consider the graph
of Figure \ref{fig3}. It can be easily seen that no tied pair arises and 
Greedy outputs a 3-vertex set  $\{1,2,3\}$. But $\gamma(G)=2$. In general, 
if $S$ contains no tied pair, then $S$ is an independent set. 
However, the non-existence of a tight pair combined with another similar
condition guarantees the optimality of Greedy,  as shown below.\\ 

\begin{figure}[H]
\centering
\begin{tikzpicture}[scale=.65, transform shape]

\node [draw, shape=circle] (s1) at  (0,0) {4};
\node [draw, shape=circle] (s2) at  (3,0) {5};
\node [draw, shape=circle] (s3) at  (6,0) {6};

\draw(s1)--(s2)--(s3);

\node [draw, shape=circle] (l1) at  (0,2) {1};
\node [draw, shape=circle] (l2) at  (3,7) {3};
\node [draw, shape=circle] (l3) at  (6,2) {2};

\draw(s1)--(l1)--(s2);
\draw(s3)--(l3)--(s2);

\node [draw, shape=circle] (t1) at  (-3,3) {};
\node [draw, shape=circle] (t2) at  (0,4) {};
\node [draw, shape=circle] (t3) at  (3,5) {};
\node [draw, shape=circle] (t4) at  (6,4) {};
\node [draw, shape=circle] (t5) at  (9,3) {};
\node [draw, shape=circle] (t6) at  (3,9) {};

\draw(l1)--(t1)--(l2)--(t6);

\draw(l1)--(t2)--(l2);

\draw(l1)--(t3)--(l2);
\draw(l3)--(t3);

\draw(l1)--(t4)--(l2);
\draw(l3)--(t4); 

\draw(t5)--(l2);
\draw(l3)--(t5);

\end{tikzpicture}
\caption{\small Graph $G$, where $S=\{1,2,3\}$ and the unique $\gamma$-set 
for $G$ is $\{3,5\}$.}\label{fig3}
\end{figure}
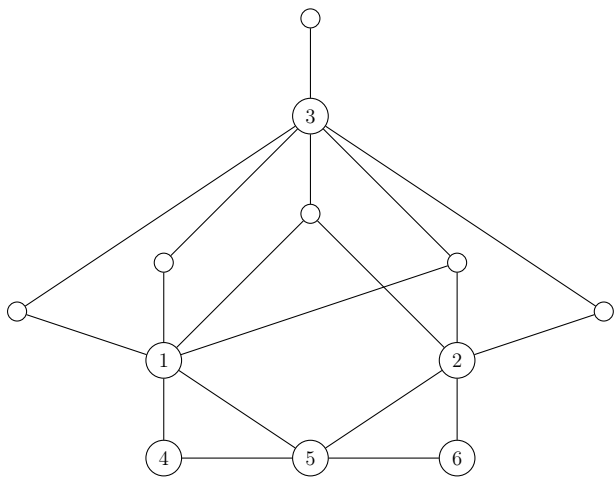

The following two observations easily follow from the definitions.\\ 

\begin{obs}\label{prop3}
	The following statements hold:
	\begin{itemize}
		\item [(i)] $ S $ is an \textit{independent set} of $ G $ if and only if $ P_S=\emptyset $;
		\item [(ii)] no vertex of $ \overline{S} $ is adjacent to (at least) two vertices of $ S $ if and only if there are
					two vertices $ v_a, v_b \in S  $ and a vertex $ w \in \overline{S} $ such that $ w \in S(a) $ and $ (w, v_b) \notin E $.
	\end{itemize}	
\end{obs}

\begin{obs}\label{prop4}
	$ S \setminus V (G(P_S)) $ is an independent set of $ G $ and every vertex of $ S \setminus V (G(P_S)) $ has some neighbor in $ V (G) \setminus S $.	
\end{obs}

Let us recall that an \textit{independent set} of a graph is a set of vertices of the graph which
are mutually non-adjacent.\\

\begin{prop}\label{prop2}
	If $ S $ is an independent set of $ G $ and if no vertex of $ \overline{S} $ is adjacent to (at least) two vertices of $ S $, then $ S $ is a minimum dominating set of $ G $.	
\end{prop}

\begin{proof}
	The assumption implies that $ \{\{v_1\} \cup N(v_1), . . . , \{v_p\} \cup N(v_p)\} $ is a partition of $ V (G) $. Then for any minimum dominating set $ D $ of $ G $ one has $ |D\cap [\{v_h\} \cup N(v_h)]|\geq 1 $ for every $ h \in \{1, . . . , p\} $. 
	It follows that $|D|\ge |S|$ and hence $ S $ must be a minimum dominating set of $ G $.
\end{proof}

\subsection{Generalized clusters and traversals}

Each tree $T_i$ represents important dependencies in the formed dominant set. Based 
on these dependencies, we may purify the resultant dominant set in different ways.  
In particular, an order, in which the vertices of each cluster are purified is 
important and affects the outcome of the purification process. In fact, a family
of algorithms can implicitly be defined according to the order in which the vertices
of each cluster are considered during the purification process. Straightforward
listing like bottom up and top down not necessarily will give good results. 
An intuitive relevant observation is that a  vertex in a cluster with a ``considerable'' 
number of immediate descendants may be left in the formed dominant set while all
its immediate descendants can be purified (unless such a descendant has a private
neighbor). Another intuitively clear observation is that, for a given vertex of a 
cluster, the number of its neighbors from set $\bar S$ matters: 
vertices with ``large'' number of such neighbors may left in the formed dominant 
set. These intuitive observations can be taken into account in different ways in 
different traversal/purification methods. The theoretical and experimental 
study of such different purification algorithms might well be the subject of another
related research.

A cluster can be seen as a dynamically formed sub-graph of graph $G$. Whenever
a pair of vertices $a,b$ with the edge $(a,b)\in E(G)$, is in a cluster, 
both $a$ and $b$ should have been included into the formed dominant set by 
Greedy. The inverse is not necessarily true, i.e., not all dependencies (edges) 
between the vertices of a cluster are present in that cluster (a cluster is a tree).   
For example, to preserve a tree-structure, we may merely include each new 
vertex as a new root of  that cluster, and there are many other possible
ways to form a cluster. If that vertex is a neighbor of two or more 
vertices from different clusters, then all these clusters will be merged
based on this dependency (again the new vertex might be declared as the
root of the new merged cluster or the merging can be accomplished in some different 
way). A chosen cluster structure affects, in general, the efficiency of 
a chosen traversal/purification method based on that cluster structure. And 
of course, for a fixed cluster structure, different traversal/purification 
methods are possible. Each of them, in general, will have different efficiency.

\section{Stage 2: Purification procedure}

In this section we describe a simple purification method based on the particular cluster structure that we defined in the previous section. Recall that 
in Section 3.1 we argued that different cluster structures and traversal methods might be exploited and tested. Based on the particular cluster structure that we defined in the previous section,  one particular traversal and purification procedure is now described.

Suppose at iteration $h$ of the purification procedure called for cluster $C_i$,
$x \in V(G) \setminus \bar S$ and $v\in T^h_i$. We call vertex $x$ 
{\em semi-private} neighbor of vertex $v$ if $v$ is the only (remaining) 
neighbor of vertex $x$ in the current purified forest 
$T^\prime_1,\dots,T^\prime_{i-1}, T^{h-1}_{i}, T_{i+1}, \dots,T_k$, 
a sub-graph of forest $T=G(P_S)$. Note that a semi-private neighbor of 
vertex $v$ is not necessarily a private neighbor of that vertex in graph $G$.

At iteration $h>1$, we distinguish three types of vertices in tree 
$T^{h-1}_i$: (1) the already purified ones; (2) the ones which were set as 
non-purified at some previous iteration -- these vertices  are called  {\em firm} 
vertices;  and (3) the remaining vertices - the ones, for which no decision yet was 
taken. The last type of vertices will be referred to as {\em pending} ones;
$PF^h_i$ is the set of pending vertices at iteration $h$. 

\begin{itemize}
	\item [(0)] {\it Initially}, at iteration 0, the purification procedure declares all leafs in tree $T_i$ having a private neighbor as firm. Then at iteration 0, it proceeds by purifying the non-firm leafs in the resultant tree. If leaf $l$ is so purified, then its parent is set as firm. The resultant tree $T^0_i$ is the purified tree of the initial iteration 0, and $PF^0_i$ is the resultant set of the vertices which remain pending.

	\item [(1a)] Before any iteration $h\ge 1$, all vertices with semi-private neighbors are set firm. Then the procedure applies one of the two purification rules as described in steps (2) and (3) below. 

	\item [(1b)] At every iteration $h\ge 1$, the {\em highest level leftmost firm vertex with a pending parent} is looked for in tree $T^{h-1}_i$. We denote this firm vertex by $a^h$. If there exists no vertex $a^h$, equivalently, there exists no pending vertex in tree $T^{h-1}_i$, then the purification procedure returns the purified tree $T^\prime_i=T^{h-1}_i$ and stops.

	\item [(2)] The {\it first purification rule} purifies two pending vertices $b$ and $c$ from path $(a,b,c,d)$ of the current tree $T^{h-1}_i$
	at iteration $h$, where $a=a^h$. We call $(a,b,c,d)$ a {\em quadruple}. Note that in such quadruple, vertex $a$ is already firm and vertices $b$ and $c$ are yet pending by iteration $h$; neither vertex $b$ nor vertex $c$ may have a semi-private neighbor by (1a). Thus vertices $b$ and $c$ are purified and vertex $d$ is set firm unless it was earlier
	 set firm.

	\item [(3)] The {\it second purification rule} is applied if there exists no quadruple, i.e., there are no two consecutive pending vertices in tree $T^{h-1}_i$. Then the procedure chooses three consecutive vertices $(a,b,c)$ on a path in that tree, such that, again,  $a=a^h$ and $b$ is a yet pending vertex by iteration $h$. We call $(a,b,c)$  a {\em trio}. Note that vertex $a$ is firm, vertex $b$ is pending and vertex $c$ could have been either set firm or purified at an earlier iteration from another branch. The second rule purifies the intermediate pending vertex $b$.

	\item [(4)] If at iteration $h$ neither the first nor the second purification rule can be applied, i.e., these exist no pending vertex in tree $T^h_i$, then  the procedure returns the purified tree $T^\prime_i=T^h_i$ and stops.

\end{itemize}  

Note that at every iteration $h$, the quadruples and trios are considered as up-going branches in tree $T^h_i$. I.e., in a quadruple $(a,b,c,d)$, $a=a^h$ is the highest level (furthest from the root) fixed vertex, and $d$ is the lowest level (the closest from the root) vertex (in particular, $a$ can be a leaf and $d$ can be the root).

\begin{algorithm*}[]
	\fontsize{10pt}{.3cm}\selectfont
	\caption{Procedure Purify of $T_i$. }\label{alg_purification_stage}
	\begin{algorithmic}
		\State Input: Tree $T_i$ and a graph $ G $.
		\State Output: The purified tree $T^\prime_i $. 
		\State $T^0_i := T_i$;
		\State $PF^0_i := V(T_i)$;
		\For {any leaf $l \in T_i$}
		\If {$l$ has no semi-private neighbor} \hspace{.3cm} \{ set $l$ 
		as purified and set its parent-node $l^-$ firm \} 
		\State $T^0_i := T^0_i\setminus \{l\}$;
		\State $PF^0_i := PF^0_i\setminus \{l,l^-\}$;
		\Else  \hspace{.3cm} \{ set $l$ firm \}  
		\State $PF^0_i := PF^0_i\setminus \{l\}$;
		\EndIf 
		\EndFor
		
		
		Iterative step: 
		\State $h := 0$;
		\While{$PF^h_i\neq\emptyset$ } 
		\State $h := h+1$;
		
		\For{any vertex $v \in PF^h_i$} 
		\If {$x$ has a semi-private neighbor} 
		\State $PF^h_i := PF^{h-1}_i\setminus \{x\}$; \hspace{.3cm} \{ set $x$ firm \}
		\EndIf
		\EndFor
		
		\State $a^h:=$ the highest level leftmost firm vertex with a pending parent in tree $T^{h-1}_i$
		\If {there exists no vertex $a^h \in T^{h-1}_i$} 
		\State return $T^\prime_i := T^{h-1}_i$ and stop;
		\EndIf
		
		\{ there exists vertex $a^h$ \}
		\If{there exists quadruple $(a^h,b,c,d)$ } \hspace{.3cm}\{Note that in such quadruple, vertices $b$ and $c$ are in $PF^{h-1}_i$\} 
		\State $T^h_i := T^{h-1}_i\setminus \{b,c\}$;
		\State $PF^h_i := PF^{h-1}_i\setminus \{b,c,d\}$;
		
		\Else \{there exists no quadruple\}
		\If{there exists trio $(a^h,b,c)$} \hspace{.3cm}\{Note that in such trio, vertex $b \in PF^{h-1}_i$\}
		\State $T^h_i := T^{h-1}_i\setminus \{b\}$;
		\State $PF^h_i := PF^{h-1}_i\setminus \{b\}$;	
		\EndIf 
		\EndIf
		
		\State return $T^\prime_i := T^{h}_i$ and stop;
		
		\EndWhile
		
	\end{algorithmic}
	
\end{algorithm*}

As we show in Lemma \ref{firm} below, a vertex, once purified at iteration $h$ from vertex $a^h$, will not be set firm in another up-going chain at a consequent iteration. Hence, no purified vertex will appear in the returned by the procedure dominant set:\\

\begin{lemma}\label{firm}
If Algorithm 2 purifies vertex $b\in T^{h-1}_i$  
at iteration $h$ then that vertex cannot be set firm at any later iteration.
\end{lemma}
\begin{proof}
Let vertex $b$ be purified from vertex $a^h$ at iteration $h$. By way of contradiction, suppose vertex $b$ is set firm from vertex $a^g$ at iteration $g>h$. By the construction of the procedure, there must exist two intermediate vertices between vertices $a^g$ and $b$ , as otherwise vertex $b$ would not have been set firm. But then the level of vertex $a^g$ is greater than that of vertex $a^h$ and hence the procedure could not select vertex $a^h$ at iteration
$h$ as the highest level firm vertex.
\end{proof}

\begin{lemma}\label{purification}
For any cluster $C_i$ and for any purified vertex $p\in T_i$, Algorithm 2 leaves either a son or the parent of vertex $p$ in the resultant purified tree $T^\prime_i$. Moreover, by eliminating the latter vertex from tree $T^\prime_i$, vertex $p$ will be left uncovered. Therefore, if the set of vertices delivered by the procedure is a dominating set, then it is also a minimal dominating set. 
\end{lemma}

\begin{proof}The first claim follows directly from the construction. We prove
the second one. Every purified vertex from tree $T_i$ is either the parent
or a son of a firm vertex. Consider two adjacent purified vertices, say
$b$ and $c$, so that $b$ is the parent of a firm vertex $a=a^h$ and $c$ is a son
of another firm vertex $d$; here $d$ is the closest one to the root, and
$(a,b,c,d)$ is a (directed) path in tree $T_i$ in the corresponding up-going 
branch. If vertex $a$ is a leaf then it has a semi-private neighbor, hence if 
vertex $a$ is purified then the resultant set will not be dominant. 
If vertex $a$ has a son which is a leaf, then that leaf is purified, hence if 
vertex $a$ is also purified then the resultant set will not be dominant again. 
If none of the sons of vertex $a$ is a leaf, then again, at least one
of them, say $x$, should have been purified as otherwise vertex $a$ could not have 
been set as a firm vertex. Moreover, a son of vertex $x$ should have also been
purified as otherwise vertex $a$ would also have been purified. 
It follows that vertex $a$ cannot be purified without leaving 
some purified vertex uncovered. Similar line of reasoning can be applied to vertex 
$d$, which proves the second claim. The third one obviously follows.
\end{proof}

A formal description of our overall two-stage algorithm  is given below.

\medskip

\begin{algorithm*}
	\caption{The overall algorithm}\label{algoritmo2}
	\begin{algorithmic}
		\State \textbf{Input:} graph $ G $ 
		\State \textbf{Output:} a dominating set $ S^* $ of graph $ G $, 
		 $ S^* \subseteq S $. \\
		
		\State Call Algorithm Stage-One to form dominating set $ S $ of graph $ G $ 
		
		\State Construct the forest $ G(P_S) $  consisting of the trees 
		$ T_1, \ldots , T_k $; 
		
		\State  $ S^* := S \setminus V (G(P_S)) $; 
				
		\For{$ i = 1, \ldots, k $} 
		
		\State Call $ Procedure\;Purify (T_i,G) $ to form purified trees $ T^\prime_i $ ;
		\State $ S^*:= S^*\cup V(T^\prime_i)$ ;

		\EndFor
		
		\State \textbf{Output} $ S^*$ 
	\end{algorithmic}

\end{algorithm*}

\medskip
\newpage

\subsection{Implementation issues}

We represent  the vertices from forest $T$ adjacent to vertex $x\in V(G)$ in a balanced binary search tree $l^0(x)$ (the nodes of these trees are ordered according to their indices). We keep  each vertex  as a pointer record data structure. Whenever vertex $v\in V(T^h_i)$ is purified at iteration $h$, we delete the node representing that vertex from every binary search tree containing a node representing vertex $v$. In this way, at every iteration $h$, for every vertex $x\in \overline {S}$, the number of the nodes in tree $l^h(x)$ equals to the number of the non-purified neighbors of vertex $x$. The number of nodes in these trees is gradually decreasing as new vertices get purified. At the same time, by the construction, for no iteration $h$ and no vertex $x$, tree $l^h(x)$ may remain empty.

It is a known fact that the number of vertices in a minimum dominating set is bounded 
above by  $\frac{n}{2}$, i.e., $\gamma(G)\leq \frac{n}{2}$, for every simple graph 
$G$ of order $n$, and this bound is tight (see the example in Figure \ref{fig_corona}
below). In Theorem \ref{t-lemma-upper-bound}, we prove that the same bound 
is valid for the dominating set delivered by the overall algorithm (Algorithm \ref{algoritmo2}). We first give
the example and then present Theorem \ref{t-lemma-upper-bound}.

The corona product of graphs $G_1$ and $G_2$, $G_1 \odot G_2$ is defined as the graph obtained 
from graphs $G_1$ and $G_2$ by taking one copy of $G_1$ and $n_1$ copies of $G_2$ and joining 
by an edge each vertex from the $i^{th}$-copy of $G_2$ with the $i^{th}$-vertex of $G_1$ ($n_1$ 
is the order of $G_1$). There are known corona product graphs of the form $G = H \odot K_1$, for 
which the upper bound $\frac{n}{2}$ is tight; here $H$ is an arbitrary graph and $K_1$ is a 
complete graph of order one. It can be easily seen that  Algorithm \ref{algoritmo2} attains 
this upper bound for such graphs.\\

\begin{figure}[h!]
\centering
\begin{tikzpicture}[scale=.9, transform shape]

\node [draw, shape=circle] (c1) at  (0,0) {};
\node [draw, shape=circle] (c2) at  (2,0) {};
\node [draw, shape=circle] (c3) at  (2,2) {};
\node [draw, shape=circle] (c4) at  (0,2) {};
\node [draw, shape=circle] (c5) at  (1,3) {};

\node [draw, shape=circle] (k1) at  (-1,1) {};
\node [draw, shape=circle] (k2) at  (3,1) {};
\node [draw, shape=circle] (k3) at  (3,3) {};
\node [draw, shape=circle] (k4) at  (-1,3) {};
\node [draw, shape=circle] (k5) at  (1,4) {};

\draw(c1)--(c2)--(c3)--(c5)--(c4)--(c1);
\draw(k1)--(c1);
\draw(k2)--(c2);
\draw(k3)--(c3);
\draw(k4)--(c4);
\draw(k5)--(c5);

\end{tikzpicture}
\caption{\small Graph $C_5 \odot K_1$}\label{fig_corona}
\end{figure}
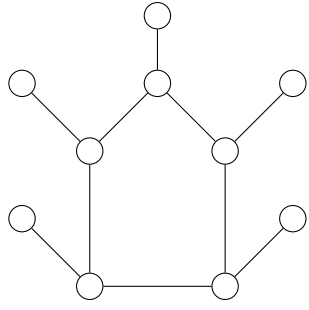
\newpage

\begin{theorem}\label{t-lemma-upper-bound}
The purification procedure delivers a minimal dominating set $S^*$ in time $O(n^2 \log n)$ and $|S^*|\leq \frac{n}{2}$.  
\end{theorem} 

\begin{proof}
Let $ G = (V, E) $ be a graph and $ A, B \subseteq V  $ with $ A \cap B = \emptyset $. Let us say that $ A = \{a_1, \ldots , a_r\} $ (for some integer $ r $) has a \textit{system of representatives in} $ B $ if there exists $ \{b_1, \ldots , b_r\} \subseteq B  $ such that $ (a_i, b_i) \in E  $ for every $ i \in \{1, \ldots , r\} $. Let us recall that Hall’s Theorem states that $ A $ has a system of representatives in $ B $ if and only if for any $ A^{\prime} \subseteq A $ one has $ |A^{\prime}| \leq |N(A^{\prime}) \cap B| $.

Let us write $ S^\prime = S \setminus V (G(P_S)) $. Then, by Observation \ref{prop4}, $S^\prime$ is an independent set of $ G $ and every vertex of $S^\prime$ has some neighbor in $ V (G) \setminus S $ [i.e. in $ N(S^\prime) \setminus S $]. \\

\textit{Claim 1.} $S^\prime$ has a system of representatives in $ N(S^\prime) \setminus S $.\\

\textit{Proof.} For brevity let us write $ N^\prime(S^\prime) = N(S^\prime) \setminus S $.\\
By contradiction, assume that $S^\prime$ has not a system of representatives in $N^\prime(S^\prime)$. Then by Hall’s Theorem there is a subset $ A^\prime \subseteq S^\prime  $ such that $ |A^\prime| > |N(A^\prime) \cap N^\prime(S^\prime)| $. Let us assume without loss of generality that $ A^\prime $ is \textit{minimal} with respect to the above property. Then let $ a^\prime $ be the last vertex of $ A^\prime $ which [in Algorithm Stage-One] is added to the greedy solution, say at iteration $ h $.\\
Let us consider the active degree of $ a^\prime $ at iteration $ h $. Note that, because of the minimality of $ A^\prime $, every vertex of $ N(A^\prime) \cap N^\prime(S^\prime) $ has a neighbor in $ A^\prime \setminus \{a^\prime\} $, i.e., is covered by $ S_{h-1} $ and does not increase the active degree of $ a^\prime $. Then let $ b^\prime $ be another possible neighbor of $ a^\prime $ with $ b^\prime \notin N^\prime(S^\prime) $: if $b^\prime \in S_{h-1} $,  then it does not increase the active degree of $ a^\prime $; if $b^\prime \notin S_{h-1} $, then by construction $ b^\prime \in S \setminus S_{h-1} $, but since $ a^\prime \in V (G(P_S))  $ and
$ b^\prime \notin V (G(P_S))  $ one has by Proposition \ref{tight} that $ b^\prime $ is covered by some vertex of $ S_{h-1} $ and does not increase the active degree of $ a^\prime $. Then the active degree of $ a^\prime $ is equal to $ 0 $. \\
On the other hand, at iteration $ h $, all neighbors of $ a^\prime $ in $ N(A^\prime) \cap N^\prime(S^\prime) $ are in $ S_{h-1} $ [in fact otherwise $ a^\prime $ would be covered by $ S_{h-1} $ and then, having active degree equal to $0$, would be not added to the greedy solution]:  then $ a^\prime $ is uncovered by $ S_{h-1} $, so that the active degree of any neighbor of $ a^\prime $ in $ N(A^\prime) \cap N^\prime(S^\prime) $ is at least 1, because of $ a^\prime $. But this is a contradiction since, at iteration $ h $, vertex $ a^\prime $ is chosen by Algorithm Stage-One. $\Diamond$\\

Finally let us observe that, since $ S^* $ has a system of representatives in $ V (G) \setminus S^* $, one has that $ S^* $ is a minimal dominating set of $ G $ and is such that $ |S^*| \leq |V (G)|/2 $.

As to the time complexity, let $n_1$ be the number of vertices in 
$V\setminus V(T)$ and $n_2$ be the number of vertices in forest $T$, 
respectively ($n_1+n_2\le n$).  The number of trees $l^{h-1}(x)$, 
$x\in V\setminus V(T_1\cup\dots\cup T_k)$, is $n_1$, 
and the number of vertices in each tree is bounded by $n_2$. At an iteration $h$: 
(i) vertex $a^h$ can be determined in time $O(n_2\log n_2)$; (ii) time $O(\log n_2)$ is 
required to locate and delete the vertex representing a purified vertex $v$ in 
each of the $n_1$ binary trees, whereas there are at most two purified vertices. 
Hence, the total cost at iteration $h$ is $O(n\log n)$. The total number of 
iterations (for purifying all the trees $T_1,\dots, T_k$), roughly, is 
$\frac{n_2}{3}$. It follows that the time complexity of the purification procedure is 
$O(n^2\log n)$.  
\end{proof}

As a result of the above results, we immediately obtain the next corollary:\\ 

\begin{cor}
	The overall algorithm is correct and can be executed in time 
	$ O(n^3) $.
\end{cor}

\section{The comparative study and the approximation ratios}

In this section we give a comparison, in terms of the
time complexity and approximation ratio of our  algorithm with the state-of-the-art 
polynomial-time approximation algorithms mentioned in the introduction.

\subsection{Comparative study}

A heuristic described in Parekh in \cite{parekh}, iteratively, extends the 
currently formed set of vertices with a vertex that covers the maximum number 
of yet uncovered vertices until a dominating set is obtained. As earlier noted, 
this algorithm is essentially the same as Greedy Algorithm. The paper gives 
neither the worst-case performance ratio nor the time complexity analysis. But
it gives  the bound $n+1-\sqrt{2m+1}$  on the total 
number of vertices in the generated dominating set. It can be
readily verified that the proof of this bound in \cite{parekh} is essentially based 
on the fact that the  edges incident in the pending vertices of each 
selected vertex $v_{h}$  are repeatedly  eliminated in the 
heuristic resulting in term $-\sqrt{2m+1}$. The same argument can be applied to 
the sequence of our graphs $G_0,G_1,\dots,G_{h_{\max}}$ in Greedy Algorithm. 

The algorithm in \cite{parekh} leaves redundant vertices in the resultant dominating 
set, similarly to Greedy Algorithm, which is not, in general, minimal. We 
illustrate the performance of the former algorithm and our extended algorithm 
on the following simple examples.  In Figure \ref{fig_1991} a graph with $n=15$ 
vertices  and $m=24$ edges is depicted. 

The solution generated by the algorithm in \cite{parekh} for this graph
consists of nine vertices 
(filled ones in Figure \ref{fig_1991}), whereas our extended algorithm obtains an 
optimal dominating set with six vertices (bold ones in Figure \ref{fig_1991}).
For the example of Figure \ref{fig_1991}B, the algorithm in \cite{parekh} obtains a
dominating set with cardinality 13, whereas our algorithm obtains again an optimal
solution with 9 vertices (note that, using  in these examples, the corresponding 
families of graphs can easily be constructed).

 \begin{figure}[h]
\centering
\begin{tikzpicture}[scale=.48, transform shape]

\node[draw] at (-4,3) {A};
\node[draw] at (6,3) {B};

\node [draw, shape=circle, fill=black!20] (s8) at  (1,4) {};
\node [draw, shape=circle, fill=black!20] (s9) at  (-1,4) {};
\node [draw, shape=circle, fill=black!20] (s1) at  (0,4) {};

\node [draw, shape=circle, fill=black] (s2) at  (2,0) {};
\node [draw, shape=circle, fill=black] (s3) at  (0,1) {};
\node [draw, shape=circle, fill=black] (s4) at  (-2,0) {};
\node [draw, shape=circle, fill=black] (s5) at  (-2,-2) {};
\node [draw, shape=circle, fill=black] (s6) at  (0,-3) {};
\node [draw, shape=circle, fill=black] (s7) at  (2,-2) {};

\node [draw, shape=circle] (l2) at  (3,1) {};
\node [draw, shape=circle] (l3) at  (0,2.5) {};
\node [draw, shape=circle] (l4) at  (-3,1) {};
\node [draw, shape=circle] (l5) at  (-3,-3) {};
\node [draw, shape=circle] (l6) at  (0,-4.5) {};
\node [draw, shape=circle] (l7) at  (3,-3) {};

\draw(s2)--(l2);
\draw(s3)--(l3);
\draw(s4)--(l4);
\draw(s5)--(l5);
\draw(s6)--(l6);
\draw(s7)--(l7);
\draw(s2)--(s3)--(s4)--(s5)--(s6)--(s7)--(s2);

\draw(s1)--(s2); \draw(s1)--(s4);
\draw(s1)--(s5); \draw(s1)--(s7);

\draw(s8)--(s2); \draw(s8)--(s3);
\draw(s8)--(s6); \draw(s8)--(s7);

\draw(s9)--(s3); \draw(s9)--(s4);
\draw(s9)--(s5); \draw(s9)--(s6);

\node [draw, shape=circle] (r1) at  (6,-3) {};
\node [draw, shape=circle] (r2) at  (8,-3) {};
\node [draw, shape=circle] (r3) at  (10,-3) {};
\node [draw, shape=circle] (r4) at  (12,-3) {};
\node [draw, shape=circle] (r5) at  (14,-3) {};
\node [draw, shape=circle] (r6) at  (16,-3) {};
\node [draw, shape=circle] (r7) at  (18,-3) {};
\node [draw, shape=circle] (r8) at  (20,-3) {};

\node [draw, shape=circle, fill=black] (p1) at  (6,-1) {};
\node [draw, shape=circle, fill=black] (p2) at  (8,-1) {};
\node [draw, shape=circle, fill=black] (p3) at  (10,-1) {};
\node [draw, shape=circle, fill=black] (p4) at  (12,-1) {};
\node [draw, shape=circle, fill=black] (p5) at  (14,-1) {};
\node [draw, shape=circle, fill=black] (p6) at  (16,-1) {};
\node [draw, shape=circle, fill=black] (p7) at  (18,-1) {};
\node [draw, shape=circle, fill=black] (p8) at  (20,-1) {};

\node [draw, shape=circle, fill=black!20] (q1) at  (7,1) {};
\node [draw, shape=circle, fill=black!20] (q2) at  (11,1) {};
\node [draw, shape=circle, fill=black!20] (q3) at  (15,1) {};
\node [draw, shape=circle, fill=black!20] (q4) at  (19,1) {};

\node [draw, shape=circle, fill=black] (x1) at  (13,4) {};

\draw(r1)--(p1)--(q1)--(x1);
\draw(r2)--(p2)--(q1);
\draw(r3)--(p3)--(q2)--(x1);
\draw(r4)--(p4)--(q2);
\draw(r5)--(p5)--(q3)--(x1);
\draw(r6)--(p6)--(q3);
\draw(r7)--(p7)--(q4)--(x1);
\draw(r8)--(p8)--(q4);

\end{tikzpicture}
\caption{\small The solution by algorithm in \cite{parekh} is represented by the
filled vertices, whereas that of  Algorithm \ref{algoritmo2}  is represented by the
dark ones.}\label{fig_1991}
\end{figure}
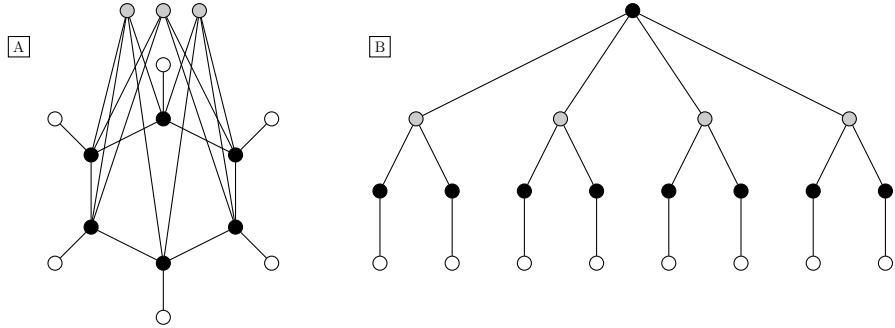

Campan \emph{et al.}  \cite{campan} compare practical performance of five polynomial approximation algorithms, three of them were proposed earlier in Eubank
\emph{et al}. \cite{eubank}. Based on this experimental study, one of the 
algorithms, referred to  as \textbf{Alg. 4}, turned out to give the best results.
Neither the approximation ratio nor time complexity of any of the considered
algorithms was shown. Nevertheless, we again compare the performance of the latter algorithm and our Algorithm \ref{algoritmo2} on simple examples. 

We form a family graphs 
$\mathcal{T}(G,H)$ as follows. Initially, we do the corona product $G \odot H$ of 
two arbitrary graphs $G$ and $H$. Then, for each vertex $v\in G$ we add a path 
graph $P_2$ of vertices $u$ and $w$ introducing edge $(v,u)$.
We show  graph $\mathcal{T}(P_2,P_4)$ in Figure \ref{fig4}A. 
For a graph from this family, \textbf{Alg. 4} obtains a dominating set with 
cardinality $n(t + 1)$, where $n$ is the order of $G$ and $t$ is the cardinality 
of the solution given by \textbf{Alg. 4} for graph $H$. 
Algorithm \ref{algoritmo2} obtains an optimal dominating set of cardinality 
$2n \leq n(\gamma(H) + 1) \leq n(t + 1)$. 

We form another  family of graphs $\mathcal{T'}(G,N_p)$ from graph $G$ of order 
$n$ and null graph $N_p$ of order $p$, $p > n$. We connect each vertex in $G$ 
with all vertices in $N_p$ by an edge, and for each vertex $v'\in G$ a path graph $P_2$ 
of vertices $u'$ and $w'$ is added by introducing edge $(v',u')$. Note that the 
resultant graph is of order $3n+p$. We illustrate graph $\mathcal{T'}(P_2,N_4)$ from 
family $\mathcal{T'}(G,N_p)$ in Figure \ref{fig4}B. 
For a graph of order $n$ from the
family,  \textbf{Alg. 4} obtains a dominating set with cardinality $p + n$, 
a magnitude greater than $\frac{3n+p}{2}$. Algorithm \ref{algoritmo2}  obtains a 
dominating set of cardinality $s + n \leq \frac{3n+p}{2} < p + n$, where $s$ is the cardinality of solution given by that algorithm for graph $G$ (recall that $G$ is 
an arbitrary graph used to form family $\mathcal{T'}(G,N_p)$).

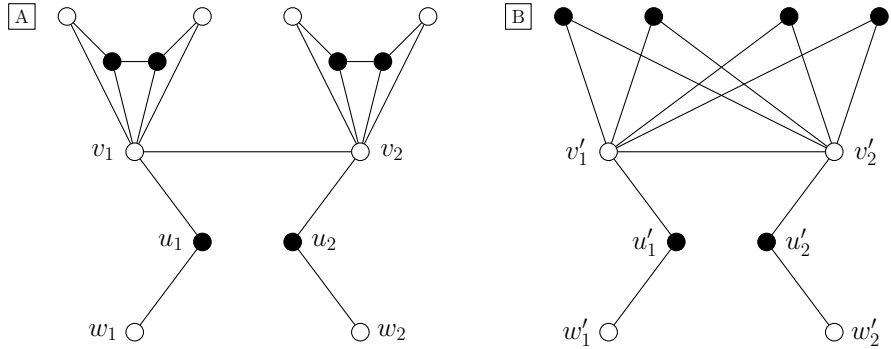
\begin{figure}[h]
\centering
\begin{tikzpicture}[scale=.6, transform shape]

\node[draw] at (-4,3) {A};
\node[draw] at (7,3) {B};

\node [draw, shape=circle] (s1) at  (-1.5,0) {};
\node [font=\Large] at (-2.2,0) {$v_1$};
\node [draw, shape=circle] (s2) at  (3.5,0) {};
\node [font=\Large] at (4.2,0) {$v_2$};
\node [draw, shape=circle, fill=black] (s3) at  (0,-2) {};
\node [font=\Large] at (-0.7,-2) {$u_1$};
\node [draw, shape=circle, fill=black] (s4) at  (2,-2) {};
\node [font=\Large] at (2.7,-2) {$u_2$};
\node [draw, shape=circle] (s5) at  (-1.5,-4) {};
\node [font=\Large] at (-2.2,-4) {$w_1$};
\node [draw, shape=circle] (s6) at  (3.5,-4) {};
\node [font=\Large] at (4.2,-4) {$w_2$};

\node [draw, shape=circle, fill=black] (p1) at  (-1,2) {};
\node [draw, shape=circle] (p2) at  (0,3) {};
\node [draw, shape=circle, fill=black] (p3) at  (-2,2) {};
\node [draw, shape=circle] (p4) at  (-3,3) {};
\draw(p2)--(p1)--(p3)--(p4);

\node [draw, shape=circle] (p5) at  (5,3) {};
\node [draw, shape=circle, fill=black] (p6) at  (4,2) {};
\node [draw, shape=circle, fill=black] (p7) at  (3,2) {};
\node [draw, shape=circle] (p8) at  (2,3) {};
\draw(p5)--(p6)--(p7)--(p8);

\draw(s5)--(s3)--(s1)--(s2)--(s4)--(s6);

\draw(s1)--(p1); \draw(s1)--(p2); \draw(s1)--(p3); \draw(s1)--(p4);

\draw(s2)--(p5); \draw(s2)--(p6); \draw(s2)--(p7); \draw(s2)--(p8);

\node [draw, shape=circle] (t1) at  (9,0) {};
\node [font=\Large] at (8.3,0) {$v'_1$};
\node [draw, shape=circle] (t2) at  (14,0) {};
\node [font=\Large] at (14.7,0) {$v'_2$};
\node [draw, shape=circle, fill=black] (t3) at  (10.5,-2) {};
\node [font=\Large] at (9.8,-2) {$u'_1$};
\node [draw, shape=circle, fill=black] (t4) at  (12.5,-2) {};
\node [font=\Large] at (13.2,-2) {$u'_2$};
\node [draw, shape=circle] (t5) at  (9,-4) {};
\node [font=\Large] at (8.3,-4) {$w'_1$};
\node [draw, shape=circle] (t6) at  (14,-4) {};
\node [font=\Large] at (14.7,-4) {$w'_2$};

\node [draw, shape=circle, fill=black] (t7) at  (8,3) {};
\node [draw, shape=circle, fill=black] (t8) at  (10,3) {};
\node [draw, shape=circle, fill=black] (t9) at  (13,3) {};
\node [draw, shape=circle, fill=black] (t10) at  (15,3) {};

\draw(t5)--(t3)--(t1)--(t2)--(t4)--(t6);

\draw(t2)--(t7)--(t1);
\draw(t2)--(t8)--(t1);
\draw(t2)--(t9)--(t1);
\draw(t2)--(t10)--(t1);

\end{tikzpicture}
\caption{\small In Figure A is shown a graph from family $\mathcal{T}(G,H)$, where 
$G \cong P_2$ ($V(G)=\{v_1,v_2\}$) and $H \cong P_4$. In Figure B is shown  graph 
$\mathcal{T'}(G,N_4)$, where $G \cong P_2$ ($V(G)=\{v'_1,v'_2\}$). The bold vertices 
represent the solution created by \textbf{Alg. 4}.}\label{fig4}
\end{figure}

A commonly studied generalization of our problem is the $k$-domination  problem. 
The $k$-domination number of a graph $G(V,E)$, $\gamma_k(G)$, is the minimum cardinality 
of a set $S \subset V$ such that any vertex in $V\setminus S$ is adjacent to at least 
$k$ vertices of $S$ (for further details on this problem, see, for instance, 
\cite{caro}). Note that for $k=1$ the $1$-domination  problem  is the same as
our domination problem. Foerster \cite{k_dominating_set}, presents a heuristic 
algorithm for $k\ge 1$ and shows that its approximation  ratio is $ln(\Delta+k)+1$, 
which coincides with the bound for Greedy Algorithm (see Equation \ref{4}), i.e., the algorithm from \cite{k_dominating_set} has the 
same approximation ratio as Greedy Algorithm for the classical domination problem. 
As for the above mentioned heuristics from \cite{parekh}, the former algorithm 
does not necessarily deliver a minimal dominating set, and there are families of 
graphs for which a dominating set formed by the algorithm in \cite{k_dominating_set}
contains more than $\frac{n}{2}$ vertices (see, for instance, an
example in Figure \ref{fig_1991}B). 

Recently,  Wang \emph{et al.} \cite{wang2} and \cite{wang1} have proposed the
heuristics for the weighted version MWDS of our domination problem.
Repeatedly, a frequency based scoring function is used to choose randomly the next 
vertex that is added to the current dominating set. The  algorithm from \cite{wang2}
initially, reduces the size of the MWDS problem and then at the second stage 
constructs the first approximation to the MWDS problem by making random choices. 
At the third stage,  the output of the second stage is reduced by eliminating some
redundant vertices applying a kind of purification procedure.  Because of the
probabilistic way for the selection of each next vertex, the heuristics has
no (deterministic) worst-case approximation ratio, but according to the reported
experimental results, it still has a good practical behavior. Nevertheless, the
purification stage still leaves redundant vertices (in particular, ones with no 
private neighbors). As a simple illustrative example, in Figure \ref{fig_mwds} 
we depict a graph for which the dominating set generated by Greedy Algorithm is
$\{1,2,3\}$. Our purification procedure eliminates redundant vertex 1, which will
be left in the dominating set constructed by the algorithm in \cite{wang2}.

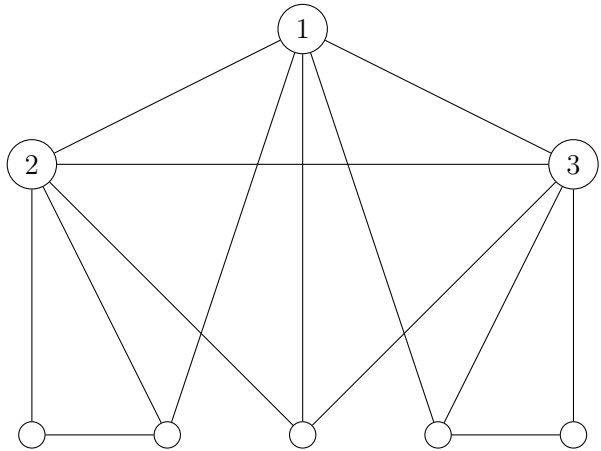
\begin{figure}[h]
\centering
\begin{tikzpicture}[scale=.9, transform shape]

\node [draw, shape=circle] (s1) at  (0,4) {2};
\node [draw, shape=circle] (s2) at  (8,4) {3};
\node [draw, shape=circle] (s3) at  (4,6) {1};

\draw(s1)--(s2)--(s3)--(s1);

\node [draw, shape=circle] (t1) at  (0,0) {};
\node [draw, shape=circle] (t2) at  (2,0) {};
\node [draw, shape=circle] (t3) at  (4,0) {};
\node [draw, shape=circle] (t4) at  (6,0) {};
\node [draw, shape=circle] (t5) at  (8,0) {};

\draw(s1)--(t1); \draw(s1)--(t2); \draw(s1)--(t3);
\draw(s2)--(t3); \draw(s2)--(t4); \draw(s2)--(t5);
\draw(s3)--(t2); \draw(s3)--(t3); \draw(s3)--(t4);

\draw(t1)--(t2); \draw(t4)--(t5);

\end{tikzpicture}
\caption{\small A graph $G$ with $\gamma(G)=2$.}\label{fig_mwds}
\end{figure}

\medskip

\subsection{Approximation ratios}

In this subsection we derive fours approximation ratios for our algorithms. 
By Remark \ref{obs-1},  all the graphs $G$ analyzed here will be such that 
$\gamma(G) > 2$. As already noted, the algorithm in \cite{k_dominating_set} obtains a k-dominating set with the approximation ratio less than $$\displaystyle\frac{ln\left(\frac{nk}{\gamma(G)}\right)}{\gamma(G)ln\left(\frac{\gamma(G)}{\gamma(G)-1}\right)}+1.$$ Now, carrying out an analysis similar to that in \cite{k_dominating_set}, we can see that, for $k=1$,  Greedy Algorithm  achieves the approximation ratio 

\begin{equation}
\rho \leq \displaystyle\frac{ln\left(\frac{n}{\gamma(G)}\right)}{\gamma(G)ln\left(\frac{\gamma(G)}{\gamma(G)-1}\right)}+1. \label{1}
\end{equation}\\

The following result is well known (see \cite{haynes1}): 

\begin{equation}
\gamma(G) \geq \frac{n}{\Delta+1}. 
\label{2}
\end{equation}

Since $\displaystyle\frac{1}{\gamma(G)ln\left(\frac{\gamma(G)}{\gamma(G)-1}\right)} \leq 1$ and (\ref{1}), we have

\begin{equation}
\rho \leq \frac{ln\left(\frac{n}{\gamma(G)}\right)}{\gamma(G)ln\left(\frac{\gamma(G)}{\gamma(G)-1}\right)} + 1 \leq ln\left(\frac{n}{\gamma(G)}\right)+1.
\label{3}
\end{equation}

 Thus, using (\ref{2}) and (\ref{3}), we obtain an approximation ratio with invariants $n$ and 
 $\Delta$

\begin{equation}
\rho \leq  ln(\Delta+1)+1.
\label{4}
\end{equation}

In the next, we derive an approximation ratio of Algorithm \ref{algoritmo2} solely in terms of the 
maximum vertex degree. Recall that $S^*$ is a minimal
dominating set returned by Algorithm \ref{algoritmo2} and that, by Theorem \ref{t-lemma-upper-bound}, 
$\frac{n}{\Delta+1} \leq |S^*| \leq \frac{n}{2}$. Since $|S^*|\geq \gamma(G)$, we obtain that 
an approximation ratio for Algorithm \ref{algoritmo2} is

\begin{equation}
\displaystyle\rho = \frac{|S^*|}{\gamma(G)} \leq \frac{\frac{n}{2}}{\frac{n}{\Delta+1}} = \frac{\Delta+1}{2}.
\label{5}
\end{equation}

Using the approximation ratios (\ref{4}) and (\ref{5}), we obtain an overall approximation ratio for
 Algorithm \ref{algoritmo2}:\\

\begin{center}
$\rho \leq \left\{\begin{array}{ll}
                                 \frac{\Delta + 1}{2}, & \mbox{if 1 $\leq \Delta \leq$ 4} \\[.3cm]
                                 ln(\Delta+1)+1, & \mbox{otherwise.}
                               \end{array}
\right.$
\end{center}

Note that, we need to verify that $\frac{\Delta + 1}{2} \leq ln(\Delta + 1)+1$. 
Let $f(\Delta)=e^{\frac{\Delta - 1}{2}} - \Delta -1$ be a function with $\Delta \in \mathbb{R}$. It is easily to see that function $f$ is 0 when $\Delta \approx 4.3567$. Since there exists a change in the sign of function $f$ when $\Delta = 4.3567$, we have that $f(\Delta) > 0$ when $\Delta > 4$. Otherwise  $f(\Delta) \leq 0$ when 
$\Delta \leq 4$.

Using a similar analysis at described above, we obtain a stronger overall bound in terms of  an optimal value $\gamma(G)$.

\begin{center}
$\rho \leq \left\{\begin{array}{ll}
                                 \frac{\Delta + 1}{2}, & \mbox{if $n \geq \gamma(G)e^{\frac{\Delta-1}{2}}$
                                 } \\[.3cm] \displaystyle ln\left(\frac{n}{\gamma(G)}\right)+1, & \mbox{otherwise.}
                               \end{array}
\right.$
\end{center}

Note that 
$$\displaystyle ln\left(\frac{n}{\gamma(G)}\right)+1 \leq ln(\Delta+1)+1.$$
 
Firstly, we need to prove that $\frac{\Delta+1}{2} \leq ln\left(\frac{n}{\gamma(G)}\right)+1$ if and only if $n \geq \gamma(G)e^{\frac{\Delta-1}{2}}$.

\begin{center}
$\displaystyle\frac{\Delta-1}{2} \leq ln\left(\frac{n}{\gamma(G)}\right)$  $\iff$ $\displaystyle e^{\frac{\Delta-1}{2}} \leq \frac{n}{\gamma(G)}$ $\iff$ 
$\displaystyle \gamma(G)e^{\frac{\Delta-1}{2}} \leq n$.
\end{center}

Consequently, if $n \geq \gamma(G)e^{\frac{\Delta-1}{2}}$, then $\rho \leq \frac{\Delta+1}{2}$, and $\rho \leq ln\left(\frac{n}{\gamma(G)}\right)+1$ otherwise.

\bigskip

Notice again that the above derived approximation ratios have an evident advantage over the 
earlier known one that it employs only invariant graph parameters and can easily be
calculated.

Finally, among the open problems arising from the analysis, the following should be highlighted.

\begin{itemize}
\item Both optimality conditions from Proposition \ref{prop3} are easily verifiable. At the same time,  it is intuitively clear that for either of the conditions  Greedy will have a performance ratio 
	essentially better than those that we obtained for the general case. Hence, it would  be important to derive tight bounds for Greedy for the two conditions. 

\item We have shown that Greedy Algorithm obtains an optimal solution for the graphs in the family 
$\mathcal{F}_1$. It would be a good idea to search for a wider class of families for which 
Algorithm \ref{algoritmo2} is optimal.

\end{itemize}

\section{Computational experiments}

In this final section we describe our computation experiments. We implemented our algorithms in C++ using Windows 10 operative system for 64 bits on a personal computer with Intel Core i7-9750H (2.6 GHz) and 16 GB of RAM DDR4. We generated the problem instances using different pseudo-random methods to generate the graphs. The order and the size of 
each  instance  was generated randomly using function $random()$ within the respective ranges. Each new edge was added in between two yet non-adjacent vertices randomly until the  corresponding size was attained.

We report our results for 835 instances which are now publicly available in \cite{bdparra}, see in Table \ref{table1}. These instances were created with the intention of verify the efficiency of  the purification stage. 
We did not include results for graphs with an excessive number of edges, 
where the corresponding dominant sets were very small compared to the number 
of nodes, hence these dominant sets were very close to the optimum (no purification was
actually required). We may observe that Stage 1 has delivered the solutions with the size essentially smaller than the known upper bounds $\frac{n}{2}$ and $n-\Delta$ on the optimum objective value. Over all the tested problem instances, in average, the reduction of about 12.24\% of the size of the dominating sets at Stage 2 was observed. Another important 
observation is that the reduction of the size of the dominating sets at the purification stage becomes more notable as the order of the graphs increases. \\

\begin{longtable}[c]{|l|l|l|l|ll|}
	\hline
	\multirow{2}{*}{No.} & \multicolumn{1}{c|}{\multirow{2}{*}{$|V(G)|$}} & \multicolumn{1}{c|}{\multirow{2}{*}{$|E(G)|$}} & \multicolumn{1}{c|}{\multirow{2}{*}{Greedy Algorithm}} & \multicolumn{2}{c|}{Overall Algorithm} \\ \cline{5-6} 
	& \multicolumn{1}{c|}{} & \multicolumn{1}{c|}{} & \multicolumn{1}{c|}{} & \multicolumn{1}{c|}{$|S^*|$} & \multicolumn{1}{c|}{Purification} \\ \hline
	\endhead
	1 & 5800 & 5849 & 2510 & \multicolumn{1}{l|}{2315} & 195 \\ \hline
	2 & 5900 & 5940 & 2550 & \multicolumn{1}{l|}{2368} & 182 \\ \hline
	3 & 6000 & 6010 & 2614 & \multicolumn{1}{l|}{2427} & 187 \\ \hline
	4 & 6050 & 6143 & 2635 & \multicolumn{1}{l|}{2441} & 194 \\ \hline
	5 & 6150 & 6237 & 2640 & \multicolumn{1}{l|}{2449} & 191 \\ \hline
	6 & 6200 & 6310 & 2694 & \multicolumn{1}{l|}{2490} & 204 \\ \hline
	7 & 6250 & 6340 & 2732 & \multicolumn{1}{l|}{2532} & 200 \\ \hline
	8 & 6300 & 6453 & 2704 & \multicolumn{1}{l|}{2504} & 200 \\ \hline
	9 & 6350 & 6491 & 2722 & \multicolumn{1}{l|}{2529} & 193 \\ \hline
	10 & 6450 & 6465 & 2841 & \multicolumn{1}{l|}{2628} & 213 \\ \hline
	11 & 6550 & 6586 & 2867 & \multicolumn{1}{l|}{2648} & 219 \\ \hline
	12 & 6700 & 6766 & 2892 & \multicolumn{1}{l|}{2666} & 226 \\ \hline
	13 & 6750 & 6870 & 2952 & \multicolumn{1}{l|}{2721} & 231 \\ \hline
	14 & 6800 & 6853 & 2945 & \multicolumn{1}{l|}{2724} & 221 \\ \hline
	15 & 6900 & 6940 & 2988 & \multicolumn{1}{l|}{2774} & 214 \\ \hline
	16 & 6950 & 6956 & 3020 & \multicolumn{1}{l|}{2807} & 213 \\ \hline
	17 & 7050 & 7142 & 3055 & \multicolumn{1}{l|}{2838} & 217 \\ \hline
	18 & 7100 & 7135 & 3067 & \multicolumn{1}{l|}{2851} & 216 \\ \hline
	19 & 7300 & 7311 & 3159 & \multicolumn{1}{l|}{2929} & 230 \\ \hline
	20 & 7350 & 7474 & 3177 & \multicolumn{1}{l|}{2944} & 233 \\ \hline
	21 & 7450 & 7497 & 3217 & \multicolumn{1}{l|}{2980} & 237 \\ \hline
	22 & 7500 & 7535 & 3270 & \multicolumn{1}{l|}{3027} & 243 \\ \hline
	23 & 7600 & 7734 & 3273 & \multicolumn{1}{l|}{3024} & 249 \\ \hline
	24 & 7650 & 7696 & 3331 & \multicolumn{1}{l|}{3095} & 236 \\ \hline
	25 & 7700 & 7716 & 3332 & \multicolumn{1}{l|}{3096} & 236 \\ \hline
	26 & 7750 & 7806 & 3343 & \multicolumn{1}{l|}{3107} & 236 \\ \hline
	27 & 7850 & 7884 & 3360 & \multicolumn{1}{l|}{3111} & 249 \\ \hline
	28 & 7900 & 7932 & 3415 & \multicolumn{1}{l|}{3161} & 254 \\ \hline
	29 & 8000 & 8126 & 3466 & \multicolumn{1}{l|}{3215} & 251 \\ \hline
	30 & 8250 & 8300 & 3612 & \multicolumn{1}{l|}{3358} & 254 \\ \hline
	31 & 8350 & 8409 & 3647 & \multicolumn{1}{l|}{3385} & 262 \\ \hline
	32 & 8400 & 8517 & 3629 & \multicolumn{1}{l|}{3360} & 269 \\ \hline
	33 & 8550 & 8606 & 3740 & \multicolumn{1}{l|}{3458} & 282 \\ \hline
	34 & 8600 & 8634 & 3711 & \multicolumn{1}{l|}{3450} & 261 \\ \hline
	35 & 8700 & 8809 & 3736 & \multicolumn{1}{l|}{3471} & 265 \\ \hline
	36 & 8800 & 8815 & 3776 & \multicolumn{1}{l|}{3504} & 272 \\ \hline
	37 & 8850 & 8864 & 3849 & \multicolumn{1}{l|}{3577} & 272 \\ \hline
	38 & 8900 & 9020 & 3840 & \multicolumn{1}{l|}{3565} & 275 \\ \hline
	39 & 8950 & 9009 & 3883 & \multicolumn{1}{l|}{3604} & 279 \\ \hline
	40 & 9100 & 9106 & 3971 & \multicolumn{1}{l|}{3691} & 280 \\ \hline
	\caption{The results for the randomly generated graphs.}
	\label{table1}
\end{longtable}

\section{Acknowledgments}

The authors are grateful to two anonymous referees. In particular, the authors
would like to thank  referee 2 for extremely extensive and extremely detailed comments 
which have lead to essential modifications in the originally submitted manuscript.   

\bibliographystyle{elsarticle-num}

\end{document}